\documentclass[a4paper]{llncs}

\usepackage[T1]{fontenc}
\usepackage{latexsym,bm}        
\usepackage{dsfont}
\usepackage{amssymb}
\usepackage[tbtags]{amsmath}
\usepackage[ruled,linesnumbered,lined,boxed,commentsnumbered]{algorithm2e}
\usepackage{color}
\usepackage{pdfsync}
\usepackage{subfigure}
\usepackage{graphicx}
\usepackage{epstopdf}
\usepackage{bibentry}

\usepackage{multirow}
\usepackage{umoline}
\definecolor{kugray5}{RGB}{224,224,224}

\graphicspath{{./figures/}}

\usepackage{url}
\urldef{\mailsa}\path|{alfred.hofmann, ursula.barth, ingrid.haas, frank.holzwarth,|
\urldef{\mailsb}\path|anna.kramer, leonie.kunz, christine.reiss, nicole.sator,|
\urldef{\mailsc}\path|erika.siebert-cole, peter.strasser, lncs}@springer.com|
\newcommand{\keywords}[1]{\par\addvspace\baselineskip
\noindent\keywordname\enspace\ignorespaces#1}

\begin{document}

\mainmatter  

\title{Exponential-Condition-Based Barrier Certificate Generation for Safety Verification of Hybrid Systems\thanks{This work was supported by the Chinese National 973 Plan under grant No.~2010CB328003, the NSF of China under grants No.~61272001, 60903030, 91218302, the Chinese National Key Technology R\&D Program under grant No.~SQ2012BAJY4052, and the Tsinghua University Initiative Scientific Research Program.}}

\titlerunning{Exponential Condition Barrier Certificate Safety Verification}

%
%
\author{Hui Kong\inst{1,5,6} \and Fei He\inst{2,5,6} \and Xiaoyu Song\inst{3} \and William N. N. Hung\inst{4} \and Ming Gu\inst{2,5,6}}
%

\institute{
Dept. of Computer Science\&Technology, Tsinghua University, Beijing, China
\and
School of Software, Tsinghua University, Beijing, China
\and
Dept. of ECE, Portland State University, Oregon, USA
\and
Synopsys Inc, Mountain View, California, USA
\and
Tsinghua National Laboratory for Information Science and Technology
\and
Key Laboratory for Information System Security, MOE, China
}

%
\maketitle

\begin{abstract}
A barrier certificate is an inductive invariant function which can be used for the safety verification of a hybrid system. Safety verification based on barrier certificate has the benefit of avoiding explicit computation of the exact reachable set which is usually intractable for nonlinear hybrid systems. In this paper, we propose a new barrier certificate condition, called \emph{Exponential Condition}, for the safety verification of semi-algebraic hybrid systems. The most important benefit of \emph{Exponential Condition} is that it has a lower conservativeness than the existing convex condition and meanwhile it possesses the property of convexity. On the one hand, a less conservative barrier certificate forms a tighter over-approximation for the reachable set and hence is able to verify critical safety properties. On the other hand, the property of convexity guarantees its solvability by semidefinite programming method. Some examples are presented to illustrate the effectiveness and practicality of our method.
\end{abstract}

\section{Introduction}
Hybrid systems~\cite{henzinger1996theory}, \cite{alur1995algorithmic} are models for those systems with interacting discrete and continuous dynamics. Embedded systems are often modeled as hybrid systems due to their involvement of both digital control software and analog plants. In recent years, as embedded systems are becoming ubiquitous, more and more researchers are devoted to the theory of hybrid systems. Reachability problems or safety verification problems are among the most challenging problems in verifying hybrid systems. The aim of safety verification is to decide that starting from an initial set, whether a continuous system or hybrid system can reach an unsafe set. For this purpose, many methods have been proposed for various hybrid systems with different features.


Deductive methods based on inductive invariant play an important role in safety verification of hybrid systems. An inductive invariant of a hybrid system is an invariant $\varphi$ that holds at the initial states of the system, and is preserved by all discrete and continuous transitions. A safety property is an invariant $\psi$ (usually not inductive) that holds in all reachable states of the system. The standard technique for proving a given property $\psi$ is to generate an inductive invariant $\varphi$ that implies $\psi$. Therefore, the problem of safety verification is converted to the problem of inductive invariant generation and hence avoid the reachability computation of the hybrid system. The key points in generating inductive invariant for hybrid systems is how to define an inductive condition that is the least conservative and how to efficiently compute the inductive invariant that satisfies the inductive condition. Usually, these two aspects contradicts with each other, that is, an inductive condition with sufficiently low conservativeness often encounters the computability or complexity problem. For different class of hybrid systems, various inductive invariants and computational methods have been proposed.

Some methods were primarily proposed for constructing inductive invariant for linear hybrid systems \cite{jirstrand1998invariant}, \cite{rodriguez2005generating}. In recent years, however, researchers concentrate more and more on nonlinear hybrid systems, especially on algebraic or semi-algebraic hybrid systems (i.e. those systems whose vector fields are polynomials and whose set descriptions are polynomial equalities or inequalities), as they have a higher universality. In \cite{sankaranarayanan2004constructing}, \cite{sankaranarayanan2010automatic}, Sankaranarayanan et al. presented a computational method based on the theory of ideal over polynomial ring and quantifier elimination for automatically generating algebraic invariants for algebraic hybrid systems. Similarly, Tiwari et al. proposed in \cite{tiwari2004nonlinear} a technique based on the theory of ideal over polynomial ring to generate the inductive invariant for nonlinear polynomial systems. In \cite{prajna2007framework}, \cite{prajna2004safety}, S. Prajna et al. proposed a new inductive invariant called \emph{Barrier Certificate} for verifying the safety of semialgebraic hybrid systems and the computational method they applied is the technique of sum-of-squares decomposition of semidefinite polynomials. In \cite{sloth2012compositional}, C. Sloth et al. proposed a new \emph{Barrier Certificate} for a special class of hybrid systems which can be modeled as an interconnection of subsystems. In \cite{platzer2008computing}, A. Platzer et al. proposed the concept of \emph{Differential Invariant} which is a boolean combination of multiple polynomial inequalities for verifying semialgebraic hybrid systems. In \cite{gulwani2008constraint}, S. Gulwani et al. proposed an inductive invariant similar to \emph{Differential Invariant} except that they defined a different inductive condition and they used SMT solver to solve the inductive invariant. In \cite{taly2009deductive}, A. Taly et al. discussed the soundness and completeness of several existing invariant condition and presented several simpler and practical invariant condition that are sound and relatively complete for different classes of inductive invariants. In \cite{taly2011synthesizing}, A. Taly et al. proposed to use inductive controlled invariant to synthesize multi-modal continuous dynamical systems satisfying a specified safety property.

In this paper, we propose a new barrier certificate (called \emph{Exponential Condition}) for the safety verification of semialgebraic hybrid systems. A barrier certificate is a special class of inductive invariant for the safety verification of hybrid systems: a function $\varphi(x)$ which maps all the states in the reachable set to non-positive reals and all the states in the unsafe set to positive reals. Given a dynamical system $S$ with dynamics $\dot{x}=f(x)$ with initial set $Init$, to prove a safety property $P$ (we use $X_u$ to denote the unsafe set) is satisfied by $S$, the basic idea of \emph{Exponential Condition} is to identify a function $\varphi(x)$ such that 1) $\varphi(x)\leq 0$ for any point $x \in Init$, 2) $\varphi(x)>0$ for any point $x \in X_u$, and 3) $\mathcal{L}_f\varphi(x) \leq \lambda\varphi(x)$, where $\mathcal{L}_f\varphi(x)=\frac{\partial{\varphi}}{\partial{x}} f(x)$ is the Lie derivative of $\varphi$ with respect to the vector field $f$ and $\lambda$ is any negative constant real value. The first condition and the third condition together guarantee that $\varphi(x) \leq 0$ for any point $x$ in the reachable set $R$, which implies that $R \cap X_u=\emptyset$. Therefore, we can assert that the safety property $P$ is satisfied by the system $M$ as long as we can find a function $\varphi(x)$ satisfying the above condition. The above condition can be extended to semialgebraic hybrid systems naturally. The idea is to identify a set of functions $\{\varphi_i(x)\}$, one for each mode of the hybrid system, which not only satisfy the above condition but also satisfy an additional sign-preserving constraint for each discrete transition.

The most important benefit of \emph{Exponential Condition} is that it is less conservative than \emph{Convex Condition} \cite{prajna2007framework} and \emph{Differential Invariant} \cite{platzer2008computing}, where the Lie derivative of $\varphi(x)$ is required to satisfy that $\mathcal{L}_f\varphi(x)\leq 0$ (a stronger condition than $\mathcal{L}_f\varphi(x) \leq \lambda\varphi(x)$), and meanwhile, it possesses the property of convexity as well. On the one hand, a less conservative inductive invariant forms a tighter over-approximation for the reachable set and hence is able to verify critical safety properties (i.e., the unsafe region is very close to reachable region). On the other hand, a convex inductive invariant condition can be solved efficiently by semidefinite programming method, which is widely used for computing Lyapunov functions in the stability analysis of nonlinear systems. In fact, there exist some other less conservative inductive invariants than \emph{Exponential Condition}, such as \cite{prajna2007framework}, \cite{gulwani2008constraint}, \cite{taly2009deductive}, however, these inductive conditions are not convex and thus cannot be solved by semidefinite programming method. Instead, they are usually solved by quantifier elimination and SMT solver, which usually has a much higher computational complexity than semidefinite programming method.

Given a semialgebraic hybrid system, we choose a set of polynomials of bounded degree with unknown coefficients as the candidate inductive invariant, and then we obtain a set of positive semidefinite polynomials (i.e. $P(x)\geq 0$) according to \emph{Exponential Condition}. Therefore, the generation of barrier certificate based on \emph{Exponential Condition} can be transformed to the problem of sum-of-squares programming of positive semidefinite polynomials~\cite{sturm1999using}, \cite{prajna2005sostools}. Based on our theory, we develop an algorithm for generating the inductive invariant satisfying \emph{Exponential Condition}. Experiments on both nonlinear systems and hybrid systems show the effectiveness and practicality of our method.

The remainder of this paper is organized as follows. Section~\ref{sec:Preliminaries} introduces the preliminaries of our method. Section~\ref{sec:CertCondition} presents the barrier certificate conditions for continuous systems and hybrid systems. Section~\ref{sec:computemethod} introduces the computational method we use to construct barrier certificates according to the barrier certificate conditions. Section~\ref{sec:examples} gives some examples to demonstrate the application of our method to the safety verification of continuous and hybrid systems. Finally, we conclude our work in Section~\ref{sec:conclusion}.

\section{Preliminaries}\label{sec:Preliminaries}
In this paper, we adopt the model proposed in~\cite{carloni2006languages} as our modeling framework. Many other models for hybrid system can be found in~\cite{maler1992prom}, \cite{lygeros1999controllers}, \cite{alur1995algorithmic}.

A continuous system is specified by a differential equation
\begin{equation}\label{sec:ContinuousSys}
  \dot{x} = f(x)
\end{equation}
where $x \in \mathbb{R}^n$ and $f$ is a Lipschitz continuous vector function from $\mathbb{R}^n$ to $\mathbb{R}^n$. Note that the Lipschitz continuity guarantees the existence and uniqueness of the solution $x(t)$ to the system~\eqref{sec:ContinuousSys}. A hybrid system can then be defined as:
\begin{definition}
\textbf{(Hybrid System)} A hybrid system is a tuple $\mathcal{H} = \langle L,X,E,R,G,$ $I,F\rangle$, where
\begin{itemize}
\item $L$ is a finite set of locations (or modes);
\item $X \subseteq \mathbb{R}^n$ is the continuous state space. The hybrid state space of the system is denoted by $\mathcal{X} = L \times X$  and a state is denoted by $(l,x) \in \mathcal{X}$;
\item $E \subseteq L \times L$ is a set of discrete transitions;
\item $G:E \mapsto 2^X$ is a guard mapping over discrete transitions;
\item $R:E \times X \mapsto 2^X$ is a reset mapping over discrete transitions;
\item $I : L \mapsto 2^X$ is an invariant mapping;
\item $F : L \mapsto (X\mapsto X)$ is a vector field mapping which assigns to each location $l$ a vector field $f$.
\end{itemize}
\end{definition}

The transition and dynamic structure of the hybrid system defines a set of trajectories. A trajectory is a sequence starting from a state $(l_0,x_0) \in \mathcal{X}_0$, where $\mathcal{X}_0 \subseteq \mathcal{X}$ is an initial set, and consisting of a series of interleaved continuous flows and discrete transitions. During the continuous flows, the system evolves following the vector field $F(l)$ at some location $l\in L$ until the invariant condition $I(l)$ is violated. At some state $(l,x)$, if there is a discrete transition $(l,l')\in E$ such that $(l,x) \in G(l,l')$ (we write $G(l,l')$ for $G((l,l'))$), then the discrete transition can be taken and the system state can be reset to $R(l,l',x)$. The problem of safety verification of a hybrid system is to prove that the hybrid system cannot reach an unsafe set $\mathcal{X}_u$ from an initial set $\mathcal{X}_0$.

An important concept used in this paper is the Lie derivative. In our context, the Lie derivative evaluates the change of a scalar function $\varphi(x)$ along the flow of a vector field $f(x)=(f_1(x),\cdots,f_n(x))$. Formally,
\begin{equation*}
  \mathcal{L}_f \varphi \triangleq \frac{\partial \varphi}{\partial x} f(x) = \sum_{i=1}^n \frac{\partial \varphi}{\partial x_i} f_i(x)
\end{equation*}

Some other notations that are used in this paper are presented here. $\mathbb{R}$ denotes the real number field. $\mathcal{C}^1(\mathbb{R}^n)$ denotes the space of 1-time continuously differentiable functions mapping $X \subseteq \mathbb{R}^n$ to $\mathbb{R}$. $\mathbb{R}[x]$ denotes the polynomial ring in $x$ over the real number field and $\mathbb{R}[x]^m$ denotes the $m$-dimensional polynomial vector space over $\mathbb{R}[x]$. $M^T$ denotes the transpose of the matrix $M$.

\section{Conditions for Constructing Barrier Certificates}\label{sec:CertCondition}
\subsection{Barrier Certificate Condition for Continuous Systems}\label{contSys}

Given a continuous system $S$, an initial set $X_0$ and an unsafe set $X_u$, a barrier certificate is a real-valued function $\varphi(x)$ of states  satisfying that $\varphi(x)\leq 0$ for any point $x$ in the reachable set $R$ and $\varphi(x)>0$ for any point $x$ in the unsafe set $X_u$ (called \emph{General Constraint} hereafter). Therefore, if there exists such a function $\varphi(x)$, we can assert that $R \cap X_u = \emptyset$, that is, the system can not reach a state in the unsafe set from the initial set. However, the exact reachable set $R$ is not computable for most hybrid systems, we cannot decide directly whether $\varphi(x)\leq 0$ holds for all the points in $R$. Therefore, various alternative inductive conditions that are equivalent to or sufficient for \emph{General Constraint} are proposed. In what follows, we present a new barrier certificate which is a sufficient condition for \emph{General Constraint}.

Consider a continuous system $\mathbb{C}$ specified by the differential equation \eqref{sec:ContinuousSys}, we assume that $X_0 (\subseteq X)$, $X_u$ are the initial set and the unsafe set respectively. Then, we have the following theorem as a barrier certificate condition.

\begin{theorem}[Exponential Condition]\label{theorem1}
  Given the continuous system \eqref{sec:ContinuousSys} and the corresponding sets $X$, $X_0$ and $X_u$, for any given $\lambda \in \mathbb{R}$, if there exists a barrier certificate, i.e, a real-valued function $\varphi(x)\in \mathcal{C}^1(\mathbb{R}^n)$ satisfying the following formulae:
\begin{align}
 &\forall x \in X_0: \varphi(x) \leq 0\label{cond11}\\
 &\forall x \in X: \mathcal{L}_f\varphi(x) - \lambda \varphi(x) \leq 0\label{cond12}\\
 &\forall x \in X_u: \varphi(x) > 0\label{cond13}
\end{align}
then the safety property is satisfied by the system \eqref{sec:ContinuousSys}.
\end{theorem}
\begin{proof}
Suppose $x_0 \in X_0$ and $x(t)$ be the corresponding particular solution of the system~\eqref{sec:ContinuousSys}. We aim to prove that for any function $\varphi(x(t))$ satisfying the formulae~\eqref{cond11}--\eqref{cond13}, the following formula holds:
\begin{equation}
  \forall \zeta \geq 0: \varphi(x(\zeta)) \leq 0.
\end{equation}
Let $g(x) = \mathcal{L}_f\varphi(x) - \lambda\varphi(x)$, then by \eqref{cond12}
\begin{equation}
  \forall x \in X: g(x) \leq 0\label{gx}
\end{equation}
Since $\frac{d\varphi(x(t))}{dt} = \frac{\partial{\varphi}}{\partial{x}} \frac{dx}{dt}=\frac{\partial{\varphi}}{\partial{x}} f(x)=\mathcal{L}_f\varphi(x)$, we have the differential equation about $\varphi(x(t))$
\begin{equation}\label{eqcomb}
\left\{\begin{array}{ll}
  \frac{d\varphi(x(t))}{dt} - \lambda \varphi(x(t)) - g(x(t)) = 0\\
  \varphi(x(0)) = \varphi(x_0)\\
    \end{array}\right.
\end{equation}
By solving the differential equation~\eqref{eqcomb}, we have following the solution:
\begin{equation}\label{odesolution}
  \varphi(x(t)) = (\int_{0}^{t}{(g(x(\tau)) e^{-\lambda \tau}d\tau}+\varphi(x_0)) e^{\lambda t}.
\end{equation}
By \eqref{gx}, we have
\begin{equation}
\int_{0}^{t}{(g(x(\tau)) e^{-\lambda \tau}d\tau} \leq 0.\label{intgx}
\end{equation}
then by \eqref{intgx} and $\varphi(x_0)\leq 0$, we finally have
\begin{equation}
  \varphi(x(t)) \leq \varphi(x_0) e^{\lambda t} \leq 0.\label{bxtineq}
\end{equation}
Hence, for any $\zeta\geq 0$, $\varphi(x(\zeta)) \leq 0$ holds.
\qed
\end{proof}

 \begin{remark}
  The formulae \eqref{cond11} and \eqref{cond13} ensure that the barrier separates the initial set $X_0$ from the unsafe set $X_u$, and the formula \eqref{cond12} ensures that system trajectories cannot escape from inside of the barrier. These conditions together imply that $\varphi(x) \leq 0$ is an inductive invariant of the system~\eqref{sec:ContinuousSys}.
  \end{remark}

  From another point of view, the semi-algebraic set $\{x\in \mathbb{R}^n| \varphi(x)\leq 0\}$ forms an over-approximation for the reachable set of the system \eqref{sec:ContinuousSys}, and the zero level set of the function $\varphi(x)$ (i.e., $\{x\in \mathbb{R}^n|\varphi(x)=0\}$) forms the boundary of the over-approximation. In order to be less conservative, we hope the boundary of the over-approximation encloses the reachable set $\{x(t)| x(0) \in X_0, \dot{x}= f(x), t \in \mathbb{R}_+\}$ as tightly as possible, in other words, to make the upper-bound of $\varphi(x(t))$ approach zero as closely as possible. According to the above proof (i.e., \eqref{bxtineq}), the scope over which the function $\varphi(x(t))$ can range depends closely on the value of the parameter $\lambda$: the less value the $\lambda$ is, the closer the upper-bound of the scope that $\varphi(x(t))$ can reach is to zero (see \figurename~\ref{fig:barrieronlambda}).
  \begin{figure}[!t]
  \centering
  \includegraphics[scale=0.5]{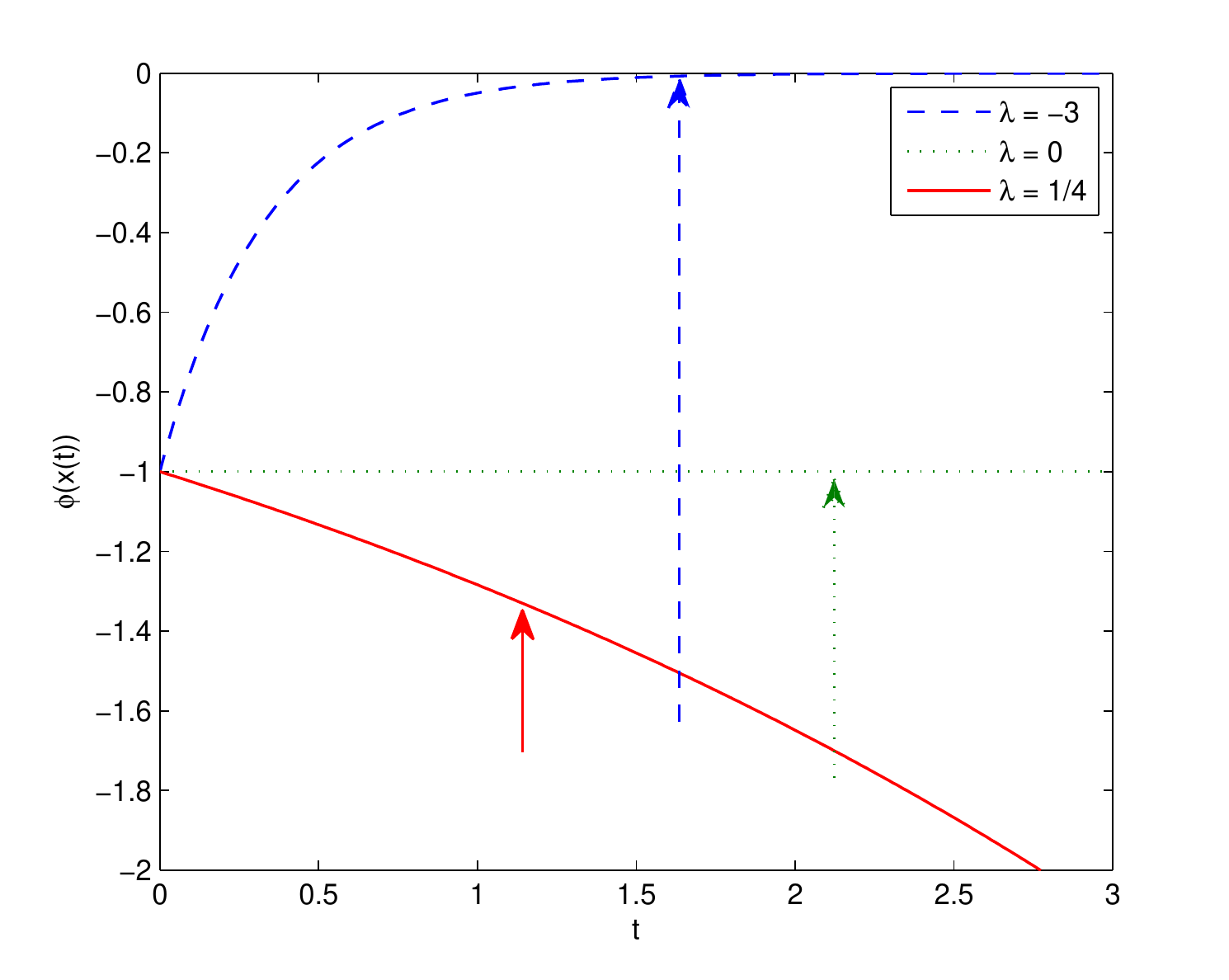}
  \caption{Dependency of Barrier Certificate Condition on $\lambda$. As the value of $\lambda$ decreases (e.g. from $1/4$ to $-3$), the upper-bound of the value of $\varphi(x(t))$ approaches to zero, which means the barrier certificate condition becomes less conservative}
   \label{fig:barrieronlambda}
\end{figure}
  Roughly speaking, the values of $\lambda$ are divided into three classes according to the conservativeness of the barrier certificate condition:
\begin{itemize}
  \item $\lambda = 0$. In this case, the formula \eqref{cond12} is degenerated to $\frac{\partial{\varphi}}{\partial{x}} f(x) \leq 0$, which is the case of \emph{Convex Condition}. This condition implies that the value of $\varphi(x(t))$ will never get close to zero over time $t$. Thus, the condition is very conservative.
  \item $\lambda <0$. In this case, we know that 1) $\varphi(x(t)) \leq \varphi(x_0) e^{\lambda t}\leq 0$, and 2) $\frac{\partial{\varphi}}{\partial{x}} f(x) \leq \lambda \varphi(x)\geq 0$. These two inequalities together imply that the value of $\varphi(x(t))$ can increase over the time $t$ but never get across the upper bound $0$, provided that $\varphi(x(0))\leq 0$ at the beginning.
  \item $\lambda >0$. In this case, $\frac{\partial{\varphi}}{\partial{x}} f(x) \leq \lambda \varphi(x)\leq 0$, which means that the value of $\varphi(x(t))$ get far away from $0$. Apparently, the condition is much more conservative than the first case.
\end{itemize}
Therefore, as long as we let $\lambda <0$, we can get less conservative barrier certificate conditions than \emph{Convex Condition}. Note that \emph{Exponential Condition} is convex as well and its convexity can be easily proved by verifying that for any two functions $\varphi_1(x)$ and $\varphi_2(x)$ satisfying the formulae~\eqref{cond11}--\eqref{cond13} and any $\theta$ with $0\leq \theta \leq 1$, $\varphi(x)= \theta \varphi_1(x) + (1-\theta) \varphi_2(x)$ satisfies the formulae~\eqref{cond11}--\eqref{cond13} as well. Based on this fact, we can convert the problem of constructing barrier certificate into the problem of convex optimization which we will discuss in Section~\ref{sec:computemethod}.

In addition, as a generalization of \emph{Convex Condition}, \emph{Differential Invariant} is basically as conservative as \emph{Convex Condition}. Here we present informally an explanation on this point. The differences in their definitions include mainly two aspects:
\begin{enumerate}
  \item invariant template: \emph{Convex Condition} employs a single inequality $p(x)\leq 0$ as the invariant template while \emph{Differential Invariant} employs a conjunction $\bigwedge_{i=1}^m q_i(x)\rhd_i r_i(x)$, where $\rhd_i$ denotes a connective in $\{=,\geq,>,\leq,<\}$.
  \item inductive condition: \emph{Convex Condition} employs $\mathcal{L}_f(p)\leq 0$ as the inductive condition while \emph{Differential Invariant} employs the conjunction $\bigwedge_{i=1}^m \mathcal{L}_f q_i\rhd_i \mathcal{L}_f r_i$, which results from applying the Lie derivative to each of the conjuncts in the invariant template respectively.
\end{enumerate}
Note that each conjunct of a \emph{Differential Invariant} is still an inductive invariant by itself, which is named \emph{Sub-Differential-Invariant} here. Based on the above definition, we can easily prove that every \emph{Sub-Differential-Invariant} $q_i(x)\rhd_i r_i(x)$ satisfies \emph{Convex Condition}. For example, suppose we have a \emph{Sub-Differential-Invariant} $q_i(x) > r_i(x)$ and the corresponding inductive condition $\mathcal{L}_f q_i>\mathcal{L}_f r_i$, let $p(x)=r_i(x) - q_i(x)$, then we can obtain an equivalent inductive invariant $p(x)<0$ and the corresponding inductive condition $\mathcal{L}_f p=\mathcal{L}_f q_i - \mathcal{L}_f r_i<0$, which implies $p(x)\leq 0$ and $\mathcal{L}_f p \leq 0$ hold. Therefore, the \emph{Sub-Differential-Invariant} $q_i(x)>r_i(x)$ satisfies \emph{Convex Condition}. Similarly, all the other cases of $q_i(x)\rhd_i r_i(x)$ can be proved to satisfy \emph{Convex Condition}. Hence, \emph{Sub-Differential-Invariant} is no less conservative than \emph{Convex Condition}. By taking a conjunction of multiple \emph{Sub-Differential Invariant}s, \emph{Differential Invariant} actually enhances the ability to over-approximate complex-shaped reachable sets. However, this does not overcome the drawback that no trajectory of the system can move towards the boundary of the over-approximation formed by a \emph{Differential Invariant}. Therefore, in this sense, we say that \emph{Differential Invariant} is basically as conservative as \emph{Convex Condition} and consequently is more conservative than \emph{Exponential Condition}.

In the following subsection, we extend the barrier certificate condition for continuous systems to hybrid systems.

\subsection{Barrier Certificate Condition for Hybrid Systems}
Different from the barrier certificate for a continuous system, the barrier certificate for a hybrid system consists of a set of functions $\{\varphi_l(x)| l \in L\}$, each of which corresponds to a discrete location of the system and forms a barrier between the reachable set and the unsafe set at that individual location. For each function $\varphi_l(x)$ at location $l$, in addition to defining constraints for the continuous flows, the barrier certificate conditions have to take into account all the discrete transitions starting from location $l$ to make the overall barrier certificate an inductive invariant. Formally, we define the barrier certificate condition for hybrid systems as the following theorem.
\begin{theorem}[Hybrid-Exp Condition]\label{theorem2}
 Given the hybrid system $\mathcal{H}=\langle L,X,$ $E,R,G,I,F\rangle$, the initial set $\mathcal{X}_0$ and the unsafe set $\mathcal{X}_u$ of $\mathcal{H}$, then, for any given set of constant real numbers $S_{\lambda}=\{\lambda_l\in \mathbb{R}| l\in L\}$ and any given set of constant non-negative real numbers $S_{\gamma}=\{\gamma_{ll'} \in \mathbb{R}_+| (l,l') \in E\}$, if there exists a set of functions $\{\varphi_l(x)| \varphi_l(x) \in \mathcal{C}^1(\mathbb{R}^n),l\in L\}$ such that, for all $l\in L$ and $(l,l')\in E$, the following conditions hold:
\begin{align}
 &\forall x \in Init(l): \varphi_l(x) \leq 0\label{cond21}\\
 &\forall x \in I(l): \mathcal{L}_{f_l}\varphi_l(x) - \lambda_l \varphi_l(x) \leq 0\label{cond22}\\
 &\forall x \in G(l,l'), \forall x' \in R((l,l'),x): \gamma_{ll'} \varphi_l(x)-\varphi_{l'}(x') \geq 0\label{cond23}\\
 &\forall x \in Unsafe(l): \varphi_l(x) > 0\label{cond24}
\end{align}
where $Init(l)$ and $Unsafe(l)$ denote respectively the initial set and the unsafe set at location $l$, then the safety property is satisfied by $\mathcal{H}$.
\end{theorem}
\begin{proof}
To prove this theorem, it is sufficient to prove that given any trajectory, say $\pi$, of the system $\mathcal{H}$, it cannot reach an unsafe state. Suppose the infinite time interval $\mathbb{R}_+$ associated with $\pi$ is divided into an infinite sequence of continuous time subintervals, i.e., $\mathbb{R}_+ = \bigcup_{n=0}^\infty I_n$, where $I_n = \{t\in \mathbb{R}_+| t_n \leq t \leq t_{n+1}\}$ is the time interval that the system spent at location $\rho(I_n)$ (where $\rho(I_n)$ returns the location corresponding to $I_n$), we define the trajectory as $\pi = \{x_{\rho(I_n)}(t)| t \in I_n, n \in \mathds{N}\}$, where $x_{\rho(I_0)}(t_0) \in Init(\rho(I_0))$. Then, our objective is to prove the following assertion:
\begin{equation}\label{assertion1}
\forall n \in \mathds{N}: \forall t\in I_n:\varphi_{\rho(I_n)}(x_{\rho(I_n)}(t))\leq 0.
\end{equation}
The basic proof idea is by induction.

\emph{Basis:}
$n=0$. According to Theorem~\ref{theorem1}, it's obvious that
\[\forall t\in I_0:\varphi_{\rho(I_0)}(x_{\rho(I_0)}(t))\leq 0\]

\emph{Induction:}
$n=k$. Assume for some $k$, \[\forall n \in [0,k]:\forall t\in I_n:\varphi_{\rho(I_n)}(x_{\rho(I_n)}(t))\leq 0\]
we mean to prove that \[\forall t\in I_{k+1}:\varphi_{\rho(I_{k+1})}(x_{\rho(I_{k+1})}(t))\leq 0\]

Case 1. (Discrete Transition) By the inductive assumption, we know that
\[\forall t\in I_k:\varphi_{\rho(I_k)}(x_{\rho(I_k)}(t))\leq 0\]
hence \[\forall t\in I_k:x(t)\in G(\rho(I_k),\rho(I_{k+1})) \implies \varphi_{\rho(I_k)}(x(t)) \leq 0\]
According to condition~\eqref{cond23}, we know that $\varphi_{\rho(I_{k+1})}(x_{\rho(I_{k+1})}(t_{k+1})) \leq 0$.

Case 2. (Continuous Transition) According to Case 1 and condition~\eqref{cond22}, we can conclude that $\forall t \in I_{k+1}:\varphi_{\rho(I_{k+1})}(x_{\rho(I_{k+1})}(t))\leq 0$ by Theorem~\ref{theorem1}.

By induction, we know that the assertion~\eqref{assertion1} holds. Therefore, the safety property is guaranteed.
\qed
\end{proof}

Informally, the formulae \eqref{cond21}, \eqref{cond22} and \eqref{cond24} together ensure that at each location $l \in L$, the system never evolves into an unsafe state continuously. The formula \eqref{cond23} ensures that the system never jumps from a safe state to an unsafe state discretely. By induction, the formulae \eqref{cond21}--\eqref{cond24} together guarantee the safety of the system.

\begin{remark}
The selection of the parameter set $S_\lambda$ is essential to the conservativeness of the barrier certificate conditions. As discussed in Subsection~\ref{contSys}, by setting all the elements of $S_\lambda$ to $0$, we can derive \emph{Convex Condition} for hybrid systems. However, \emph{Convex Condition} is too restrictive to be useful for hybrid systems. For example, see the hybrid system in \figurename~\ref{fig:convexcondexamp1}, there is a reset operation $x = x_r$ (which is often the case) at the transition $(l_2,l_1)$. Assume there exists a barrier certificate $\{\varphi_{l_1}(x),\varphi_{l_2}(x)\}$ if we set all the elements of $S_\lambda$ to $0$ and (without loss of generality) set all the elements of $S_\gamma$ to $1$, then for any trajectory containing at least two times of the transition $(l_2,l_1)$, one at time instant $t_1$ and another at $t_2$, $t_1 < t_2$, respectively, we can assert that $\varphi_{l_1}(x_{l_1t_1}) > \varphi_{l_1}(x_{l_1t_2})$ according to Theorem~\ref{theorem2}, this contradicts with $x_{l_1t_1} = x_{l_1t_2} = x_r$, that is, the barrier certificate satisfying \emph{Convex Condition} does not exist no matter what the unsafe set is.
\begin{figure}[!t]
  \centering
  \includegraphics[scale=0.45]{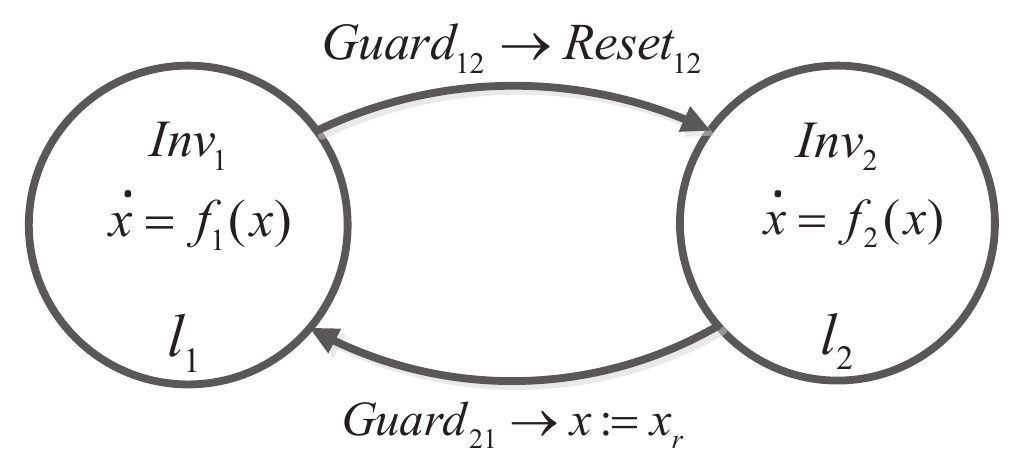}
  \caption{A hybrid system without barrier certificate satisfying \emph{Convex Condition}.}
   \label{fig:convexcondexamp1}
\end{figure}
Therefore, in order to make the barrier certificate condition less conservative, we try to choose negative values for $\lambda_l \in S_\lambda$ and theoretically: the less, the better. However, in practice, the optimal domain for $\lambda$ may depend on the specific computational method. For example, the interval $[-1,0)$ appears to be optimal and not too sensitive in-between for the semidefinite programming method used in this paper.

The selection of $S_\gamma$ is relatively simple. We usually set all of its elements to $1$ except for the discrete jumps with a reset operation that is independent of the pre-state of the jump, for which we usually set $\gamma_{ll'}$ to $0$.
\end{remark}

\section{Construction Method for Barrier Certificate}\label{sec:computemethod}
Constructing inductive invariants for general hybrid systems is very hard. Fortunately, for some existing inductive conditions, several computational methods are available for semialgebraic hybrid systems. The most representative methods include the fixed-point method based on saturation~\cite{platzer2008computing}, the constraint-solving methods based on semidefinite programming~\cite{prajna2004safety} and quantifier elimination~\cite{gulwani2008constraint} and the Grobner-bases method~\cite{tiwari2004nonlinear},~\cite{sankaranarayanan2010automatic}. Similar to \emph{Convex Condition}, \emph{Exponential Condition} defines a convex set of barrier certificate functions as well and hence can be solved by semidefinite programming method supposing the hybrid system is semialgebraic and the barrier certificate function $\varphi(x)$ is a polynomial.

In our computational method, a barrier certificate is assumed to be a set $\Phi=\{\varphi_l(x)| l\in L\}$ of multivariate polynomials of fixed degrees with a set of unknown real coefficients. According to the constraint inequalities in Theorem~\ref{theorem1} or Theorem~\ref{theorem2}, we can obtain a set of positive semidefinite (\emph{PSD}) polynomials $Q = \{Q_i| Q_i(x) \geq 0, deg(Q_i)= 2n, x \in \mathbb{R}^n, n \in \mathds{N}\}$, where $deg(\cdot)$ returns the degree of a polynomial. Note that a polynomial $Q(x)$ of degree $2k$ is said to be \emph{PSD} \emph{if and only if} $Q(x)\geq 0$ for all $x\in \mathbb{R}^n$. Thus, our objective is to find a set of real-valued coefficients for $\varphi_l\in \Phi$ to make all the $Q_i\in Q$ be \emph{PSD}.

A famous sufficient condition for a polynomial $P(x)$ of degree $2k$ to be \emph{PSD} is that it is a sum-of-squares (\emph{SOS}) $P(x) = \sum q_i(x)^2$ for some polynomials $q_i(x)$ of degree $k$ or less~\cite{lasserre2007sufficient}. Furthermore, it is equivalent to that $P(x)$ has a positive semidefinite quadratic form, i.e., $P(x) = v(x)Mv(x)^T$, where $v(x)$ is a vector of monomials with respect to $x$ of degree $k$ or less and $M$ is a real symmetric \emph{PSD} matrix with the coefficients of $P(x)$ as its entries. Therefore, the problem of finding a \emph{PSD} polynomial $P(x)$ can be converted to the problem of solving a linear matrix inequality (LMI) $M\succeq 0$~\cite{boyd1994linear}, which can be solved by semidefinite programming~\cite{parrilo2003semidefinite}.

In our work, we extend SOSTOOLS based on the theory in this paper to implement an algorithm for discovering barrier certificate automatically.

\subsection{Sum-of-squares Transformation for Continuous System}\label{subsec:continuousmethod}
In order to be solvable for the barrier certificate condition by SOS programming, we need to restate it with multivariate polynomials. In this context, we assume that all the state sets involved in the condition are semialgebraic, that is, they can be written as $\{x \in \mathbb{R}^n|P_1(x) \geq 0,...,P_m(x)\geq 0, P_i(x)\in \mathbb{R}[x], 1\leq i\leq m\}$). For convenience, we write it compactly as $\{x \in \mathbb{R}^n|\mathcal{P}(x) \geq 0, \mathcal{P}(x) \in \mathbb{R}[x]^m\}$, where $\mathcal{P}(x)=(P_1(x),P_2(x),...,P_m(x))$. In addition, each dimension of the vector field $f(x)$ and the barrier certificate function $\varphi(x)$ are all polynomials in $\mathbb{R}[x]$. Based on the previous assumption, we present the sum-of-squares transformation of \emph{Exponential Condition} for continuous systems as the following corollary.
\begin{corollary}\label{corollary1}
  Given the continuous polynomial system \eqref{sec:ContinuousSys} and the initial set $X_0 = \{x\in \mathbb{R}^n|I_0(x)\geq 0, I_0(x) \in \mathbb{R}[x]^r\}$ and the unsafe set $X_u = \{x\in \mathbb{R}^n|U(x) \geq 0, U(x) \in \mathbb{R}[x]^s\}$, where $r$ and $s$ are the dimensions of the polynomial vector spaces, for any $\lambda \in \mathbb{R}$ and any real number $\epsilon>0$, if there exists a polynomial function $\varphi(x)\in \mathbb{R}[x]$ and two \emph{SOS} polynomial vectors (i.e., every element of the vector is a \emph{SOS} polynomial) $\mu(x) \in \mathbb{R}[x]^r$ and $\eta(x)\in \mathbb{R}[x]^s$ satisfying that the following polynomials
\begin{align}
 &-\varphi(x) - \mu(x) I_0(x) \label{corocond1}\\
 &-\mathcal{L}_f\varphi(x) + \lambda   \varphi(x) \label{corocond2}\\
 &\varphi(x) - \eta(x) U(x) - \epsilon \label{corocond3}
\end{align}
are all \emph{SOS}s, then the safety property is satisfied by the system \eqref{sec:ContinuousSys}.
\end{corollary}
\begin{proof}
It is sufficient to prove that any $\varphi(x)$ satisfying \eqref{corocond1}--\eqref{corocond3} also satisfies \eqref{cond11}--\eqref{cond13}. By \eqref{corocond1}, we have $-\varphi(x) - \mu(x) I_0(x)\geq 0$, that is, $\varphi(x) \leq -\mu(x) I_0(x)$. Because for any $x \in X_0$, $-\mu(x) I_0(x)\leq 0$, this means $\varphi(x) \leq 0$. Similarly, we can derive \eqref{cond12} from \eqref{corocond2}. By \eqref{corocond3}, it's easy to prove that $\varphi(x)-\epsilon \geq 0$ holds for any $x \in X_u$. Since $\epsilon$ is greater than $0$, then the condition~\eqref{cond13} holds. Therefore, the system~\eqref{sec:ContinuousSys} is safe.
\qed
\end{proof}
\begin{remark}
Since the polynomials~\eqref{corocond1}--\eqref{corocond3} are required to be \emph{SOS}s, each of them can be transformed to a positive semidefinite quadratic form $v(x)M_iv(x)^T$, where $M_i$ is a real symmetric \emph{PSD} matrix with the coefficients of $\varphi(x)$, $\mu(x)$ and $\eta(x)$ as its variables. As a result, we obtain a set of \emph{LMI}s $\{M_i \succeq 0\}$ which can be solved by semidefinite programming.
\end{remark}

\begin{algorithm}[!t]\label{algorithm1}
\caption{Computing Barrier Certificate for Continuous System}
{\footnotesize
\SetKwFunction{GenPoly}{GenPoly}\SetKwFunction{Deg}{Deg}\SetKwFunction{Diff}{Diff}\SetKwFunction{SOSProgram}{SOSProgram}
\SetKwInOut{Input}{input}\SetKwInOut{Variables}{Variables}\SetKwInOut{Output}{output}\SetKwInOut{Constants}{Constants}
\KwIn {$f$: array of polynomial vector field; $I_0$: array of polynomials defining $X_0$; $U$: array of polynomials defining $X_u$}
\KwOut{$\varphi$: barrier certificate polynomial}
\Variables{$\lambda$: a real negative value; $d$: degree of $\varphi$}
\Constants{$\Lambda$: array of candidate values for $\lambda$; $\epsilon$: a positive value; $dMin$, $dMax$: the minimal degree and maximal degree of $\varphi$ to be found}
\BlankLine
Initialize. Set $\Lambda$ to a set of negative values between $-1$ and $0$; Set $\epsilon$ to a small positive value; Set $dMin$ and $dMax$ to positive integer respectively\;
Pick $\lambda$ and $d$. For each $\lambda \in \Lambda$ and for each $d$ from $dMin$ to $dMax$, perform step \ref{loopstart}--\ref{loopend} until a barrier certificate is found\;
Decide the degree of $\mu(x)$ and $\eta(x)$ according to $d$. To be \emph{SOS}s for both \eqref{corocond1} and \eqref{corocond3}, at least one of the degrees of $\mu(x) I_0(x)$ and $\eta(x) U(x)$ is greater than or equal to the degree of $\varphi(x)$\;\label{loopstart}
Generate complete polynomials $\varphi(x)$, $\mu(x)$ and $\eta(x)$ of specified degree with unknown coefficient variables\;
Eliminate the monomials of odd top degrees in \eqref{corocond1}--\eqref{corocond3}, $\mu(x)$ and $\eta(x)$, respectively. To be a \emph{SOS}, a polynomial has to be of even degree. Concretely, let the coefficients of the monomials to be eliminated be zero to get equations about coefficient variables and then reduce the number of coefficient variables by solving the equations and substituting free variables for non-free variables in all the related polynomials\;\label{eliminatemon}
Perform the \emph{SOS} programming on the positive semidefinite constraints \eqref{corocond1}--\eqref{corocond3} and $\mu(x)$, $\eta(x)$\;\label{programming}
Check if a feasible solution is found, if not found, continue with a new loop; else, check if the solution can indeed enable the corresponding polynomials to be \emph{SOS}s, if so, return $\varphi(x)$; else, for all the polynomials in the programming, eliminate all the monomials whose coefficients have too small absolute values(usually less than $10^{-5}$) by using the same method as step~\ref{eliminatemon}, then go to step~\ref{programming} unless an empty polynomial is produced\;\label{loopend}}
\end{algorithm}

We use \emph{Algorithm}~\ref{algorithm1} to compute the desired barrier certificate. In the algorithm, we first choose a small set of negative values $\Lambda$ as a candidate set for $\lambda$ and an integer interval $[dMin,dMax]$ as a candidate set for degree $d$ of $\varphi(x)$. Then, we attempt to find a barrier certificate satisfying the conditions~\eqref{corocond1}--\eqref{corocond3} for a fixed pair of $\lambda$ and $d$ until such one is found. Theoretically, according to the analysis about the dependence of conservativeness of barrier certificate on the value of $\lambda$, we should set $\lambda$ to as small negative value as possible. However, experiments show that too small negative numbers for $\lambda$ often lead the semidefinite programming function to numerical problems. In practice, the negative values in the interval $[-1,0)$ are good enough for $\lambda$ to verify very critical safety properties. Note that the principle for step~\ref{loopstart} in \emph{Algorithm 1} is that if $\varphi(x)$ has a dominating degree in both polynomials, there couldn't exist a solution that make both polynomials be \emph{SOS}s because $-\varphi(x)$ and $\varphi(x)$ occur in \eqref{corocond1} and \eqref{corocond3} simultaneously. The motive for eliminating the monomials with small coefficients in step~\ref{loopend} is from the observation that those monomials are usually the cause of the failed \emph{SOS} decomposition for the polynomials when the semidefinite programming function gives a seemingly feasible solution.

The idea for constructing barrier certificates for continuous systems can be easily extended to hybrid systems. We describe it in the following subsection.

\subsection{Sum-of-squares Transformation for Hybrid System}\label{subsec:hybridmethod}
Similar to continuous system, in order to be solvable by semidefinite programming, we need to limit the hybrid system model in Section~\ref{sec:Preliminaries} to semialgebraic hybrid system.

Consider the hybrid system $\mathbb{H} = \langle L,X,E,R,G,I,F\rangle$, where the mappings $F,R,G,I$ of $\mathbb{H}$ are defined with respect to polynomial inequalities as follows:
\begin{itemize}
  \item $F:l\mapsto~f_l(x)$
  \item $G:(l,l')\mapsto\{x\in\mathbb{R}^n|G_{ll'}(x)\geq~0, G_{ll'}(x)\in\mathbb{R}[x]^{p_{ll'}}\}$
  \item $R:(l,l',x)\mapsto~\{x'\in\mathbb{R}^n|R_{ll'x}(x')\geq~0, R_{ll'x}(x')\in\mathbb{R}[x]^{q_{ll'}}\}$
  \item $I: l\mapsto\{x\in\mathbb{R}^n|I_l(x)\geq~0, I_l(x)\in\mathbb{R}[x]^{r_l}\}$
\end{itemize}
and the mappings of the initial set and the unsafe set are defined as follows:
\begin{itemize}
  \item $\operatorname{Init}:l\mapsto\{x\in\mathbb{R}^n| \operatorname{Init}_l(x)\geq 0, \operatorname{Init}_l(x)\in\mathbb{R}[x]^{s_l}\}$
  \item $\operatorname{Unsafe}:l\mapsto\{x\in\mathbb{R}^n| \operatorname{Unsafe}_l(x)\geq 0, \operatorname{Unsafe}_l(x)\in\mathbb{R}[x]^{t_l}\}$
\end{itemize}
where $p_{ll'}$, $q_{ll'}$, $r_l$, $s_l$ and $t_l$ is the dimension of polynomial vector space.
Then we have the following corollary for constructing barrier certificate for the semialgebraic hybrid system $\mathbb{H}$.
\begin{corollary}\label{corollary3}
Let the hybrid system $\mathbb{H}$ and the initial state set mapping $Init$ and the unsafe state set mapping $Unsafe$ be defined as the above. Then, for any given set of constant real numbers $S_{\lambda}=\{\lambda_l\in\mathbb{R}| l\in~L\}$ and any given set of constant non-negative real numbers $S_{\gamma}=\{\gamma_{ll'}\in\mathbb{R}_+| (l,l')\in~E\}$ ,and any given small real number $\epsilon>0$, if there exists a set of polynomial functions $\{\varphi_l(x)\in\mathbb{R}[x]| l\in~L\}$ and five sets of \emph{SOS} polynomial vectors $\{\mu_l(x)\in\mathbb{R}[x]^{s_l}| l\in~L\}$, $\{\theta_l(x)\in\mathbb{R}[x]^{r_l}| l\in~L\}$, $\{\kappa_{ll'}(x) \in \mathbb{R}[x]^{p_{ll'}}| (l,l') \in E\}$ , $\{\sigma_{ll'}(x)\in\mathbb{R}[x]^{q_{ll'}}| (l,l')\in~E\}$ and $\{\eta_l(x)\in\mathbb{R}[x]^{t_l}| l\in~L\}$, such that the polynomials
\begin{align}
&\varphi_l(x) - \mu_l(x) \operatorname{Init}_l(x)\label{contCond1}\\
&\lambda_l \varphi_l(x) - \mathcal{L}_{f_l}\varphi_l(x)- \theta_l(x) I_l(x)\label{contCond2}\\
&\gamma_{ll'} \varphi_l(x) - \varphi_{l'}(x') -\kappa_{ll'}(x) G_{ll'}(x)-\sigma_{ll'}(x') R_{ll'x}(x') \label{contCond3}\\
&\varphi_l(x) - \epsilon - \eta_l(x) \operatorname{Unsafe}_l(x)\label{contCond4}
\end{align}
are \emph{SOS}s for all $l\in L$ and $(l,l')\in E$, then the safety property is satisfied by the system $\mathbb{H}$.
\end{corollary}
\begin{proof}
Similar to Corollary~\ref{corollary1}, it's easy to prove that any set of polynomials $\{\varphi_l(x)\}$ satisfying \eqref{contCond1}--\eqref{contCond4} also satisfies \eqref{cond21}--\eqref{cond24}, hence the hybrid system $\mathbb{H}$ is safe.
\qed
\end{proof}

The algorithm for computing the barrier certificates for hybrid systems is similar to the algorithm for continuous systems except that it needs to take into account the constraint~\eqref{contCond3} for the discrete transitions. We do not elaborate on it here any more. Note that the strategy for the selection of $\lambda$'s for continuous system applies here as well and we only need to set all the elements of $S_\gamma$ to $1$ except for the discrete transition whose post-state is independent of the pre-state, where we set $\gamma_{ll'}$ to $0$ to reduce the computational complexity.


\section{Examples}\label{sec:examples}
\subsection{Example 1}\label{subsec:example1}
Consider the two-dimensional system (from \cite{khalil2002nonlinear} page 315)
\begin{gather*}\label{system1}
\begin{bmatrix}\dot{x_1} \\ \dot{x_2}\end{bmatrix} = \begin{bmatrix} x_2\\ -x_1 + \frac{1}{3}x_1^3 - x_2\end{bmatrix}
\end{gather*}
with $\mathcal{X} = \mathbb{R}^2$, we want to verify that starting from the initial set $X_0 = \{x \in \mathbb{R}^2| (x_1-1.5)^2 + x_2^2 \leq 0.25\}$, the system will never evolve into the unsafe set $X_u = \{x \in \mathbb{R}^2| (x_1 + 1)^2 + (x_2 + 1)^2 \leq 0.16\}$. We attempted to use both the method based on \emph{Convex Condition} proposed in \cite{prajna2007framework} and the method based on \emph{Exponential Condition} in this paper to find the barrier certificates with a degree ranging from $2$ to $10$. (Note that in~\cite{gulwani2008constraint}, \cite{taly2009deductive}, the inductive invariants are not sufficient in general according to~\cite{taly2011synthesizing} and hence cannot be applied to our examples. The work of [19] applies only to a very special class of hybrid systems which is not applicable to our examples either.) During this process, all the programming polynomials are complete polynomials automatically generated (instead of the non-complete polynomials consisting of painstakingly chosen terms) and all the computations are performed in the same environment. The result of the experiment is listed in Table~\ref{tbl:comparison}. The first column is the degree of the barrier certificate to be found, the second column is the amount of time spent by the method based on \emph{Convex Condition}, and the rest columns are the amount of time spent by the method based on \emph{Exponential Condition} for different value of $\lambda$. Note that the symbol $\times$ in the table indicates that the method failed to find a barrier certificate with the corresponding degree either because the semidefinite programming function found no feasible solution or because it ran into a numerical problem.
\begin{table}\label{tbl:comparison}
  \caption{Computing results for \emph{Convex Condition} and \emph{Exponential Condition}. \emph{Exponential Condition} shows much stronger capability in finding barrier certificates.}
  \centering
  \small
  \begin{tabular}{|c||c||c|c|c|}
    \hline
    \multicolumn{1}{|c||}{{Degree}}&\multicolumn{1}{c||}{Convex Condition}&\multicolumn{3}{c|}{Exponential Condition}\\
    \cline{2-5}
    of&\multirow{2}{*}{$Time$(sec)}&\multicolumn{3}{c|}{$Time$(sec)}\\
    \cline{3-5}
    $\varphi(x)$&&$\lambda=\frac{-1}{8}$&$\lambda=\frac{-1}{4}$&$\lambda=-1$\\
    \hline
    \hline
     2&$\times$&0.4867&0.4836&0.2496\\
     3&$\times$&0.5444&0.6224&0.4976\\
     4&0.4368&0.4103&0.4072&0.3853\\
     5&$\times$&0.4321&0.4103&0.3947\\
     6&$\times$&0.3214&0.3011&0.2714\\
     7&$\times$&0.9563&0.9532&0.9453\\
     8&$\times$&0.9188&0.8970&0.7893\\
     9&$\times$&1.4944&1.4149&1.5132\\
     10&$\times$&1.4336&1.3931&1.3650\\
    \hline
  \end{tabular}
\end{table}

As shown in Table~\ref{tbl:comparison}, the method based on \emph{Convex Condition} succeeded only in one case ($Degree = 4$) due to the conservativeness of \emph{Convex Condition}. Comparably, our method found all the barrier certificates of the specified degrees ranging from $2$ to $10$. Especially, the lowest degree of barrier certificate we found is quadratic: $\varphi(x) = -.86153-.87278x_1-1.1358x_2-.23944x_1^2-.5866x_1x_2$ with $\mu(x)=0.75965$ and $\eta(x)=0.73845$ when $\lambda$ is set to $-1$. The phase portrait of the system and the zero level set of $\varphi(x)$ are shown in \figurename~\ref{fig:example1}. Note that being able to find a lower degree of barrier certificates is essential in reducing the computational complexity.
\begin{figure}[!t]
  \centering
  \subfigure[Subsection \protect \ref{subsec:example1}]{
  \includegraphics[scale=0.5]{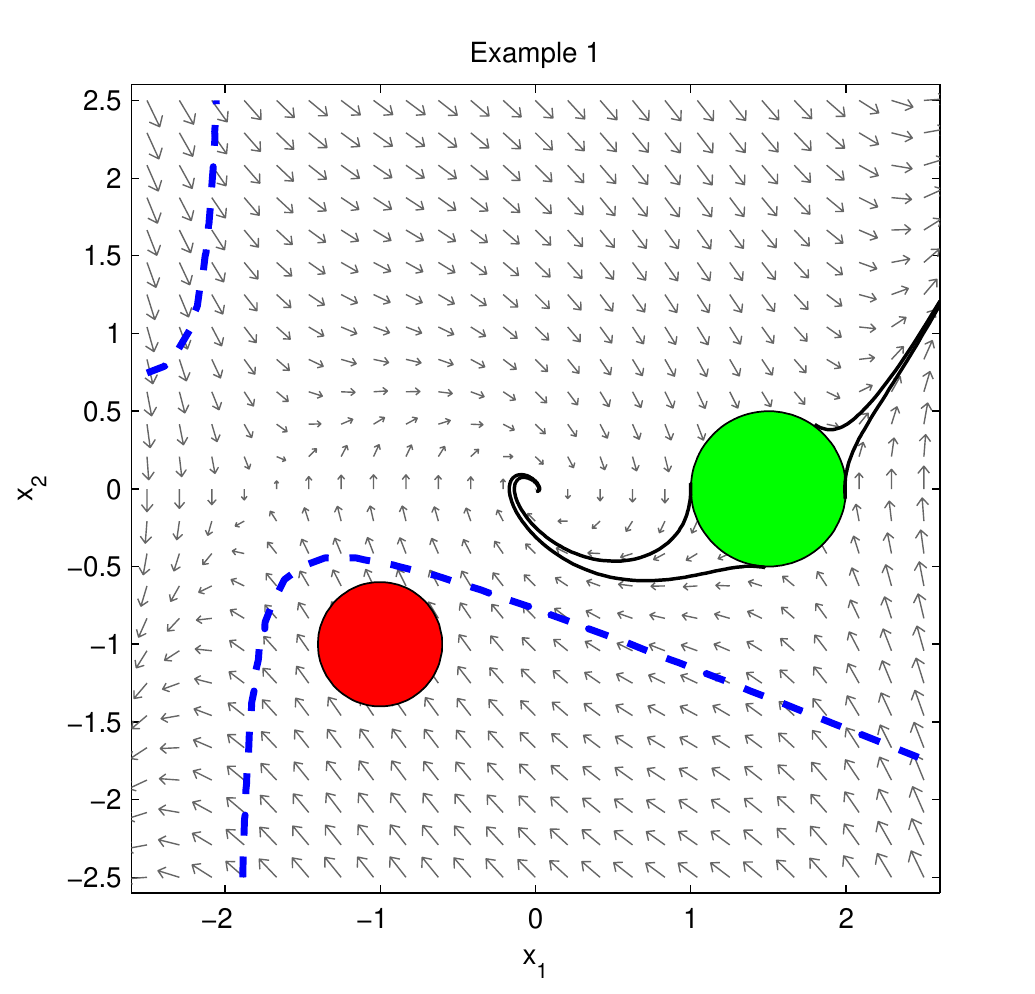}
  \label{fig:example1}
  }
  \subfigure[Subsection \protect \ref{subsec:example3}]{
  \includegraphics[scale=0.4]{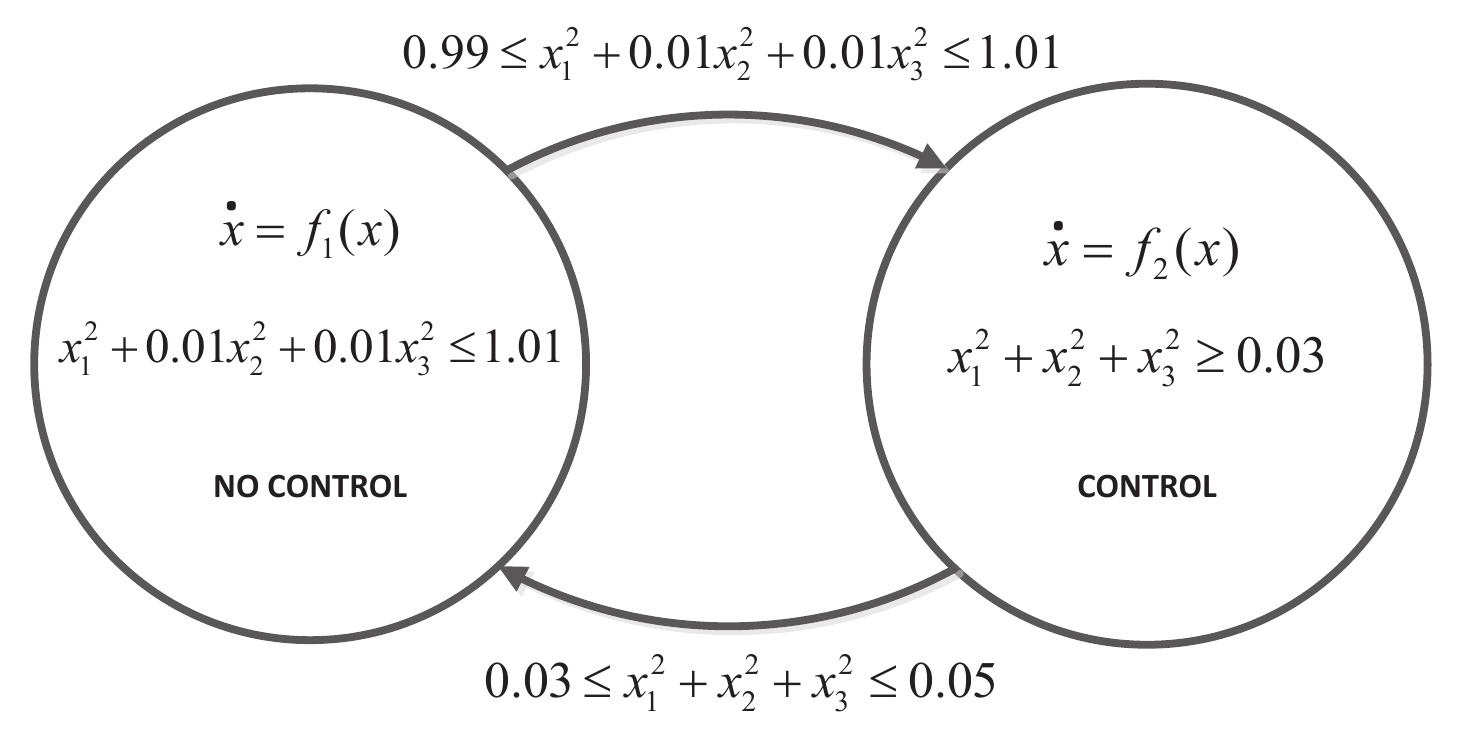}
   \label{fig:example3}
  }
  \caption{(a) Phase portrait of the system in Subsection \protect \ref{subsec:example1}. The solid patches from right to left are $X_0$ and $X_u$, respectively, the solid lines depict the boundary of the reachable region of the system from $X_0$, and the dashed lines are the zero level set of a quadratic barrier certificate $\varphi(x)$ which separates the unsafe region $X_u$ from the reachable region.
  (b) Discrete transition diagram of the hybrid system in Subsection \protect \ref{subsec:example3}.}
\end{figure}

In addition, we can see from Table~\ref{tbl:comparison} that the runtime of \emph{Exponential Condition}-based method decreases with the value of $\lambda$ for each fixed degree except for $Degree=3,9$, this observation can greatly evidence our theoretical result about $\lambda$ selection: the less, the better.

\subsection{Example 2}\label{subsec:example3}
In this example, we consider a hybrid system with two discrete locations (from~\cite{prajna2004safety}). The discrete transition diagram of the system is shown in \figurename~\ref{fig:example3} and the vector fields describing the continuous behaviors are as follows:
\begin{align*}
&f_1(x)= \begin{bmatrix}x_2 \\ -x_1+ x_3 \\ x_1+(2x_2+3x_3)(1+x_3^2)\end{bmatrix}, f_2(x)= \begin{bmatrix}x_2 \\ -x_1+x_3 \\ -x_1-2x_2-3x_3\end{bmatrix}
\end{align*}
At the beginning, the system is initialized at some point in $X_0 = \{x \in \mathbb{R}^3 | x_1^2 + x_2^2 + x_3^2 \leq 0.01\}$ and then it starts to evolve following the vector fields $f_1(x)$ at location 1(NO CONTROL mode). When the system reaches some point in the guard set $G(1,2)=\{x\in \mathbb{R}^3|0.99 \leq x_1^2 + 0.01x_2^2 + 0.01x_3^3 \leq 1.01\}$, it can jump to location 2 (CONTROL mode) nondeterministically without performing any reset operation (i.e., $R(1,2,x)=G(1,2)$). At location 2, the system will operate following the vector field $f_2(x)$, which means that a controller will take over to prevent $x_1$ from getting too big. As the system enters the guard set $G(2,1) = \{x \in \mathbb{R}^3 | 0.03 \leq x_1^2 + x_2^2 + x_3^2 \leq 0.05\}$, it will jump back to location 1 nondeterministically again without reset operation (i.e., $R(2,1,x)=G(2,1)$). Different from the experiment in~\cite{prajna2004safety}, where the objective is to verify that $|x_1| < 5.0$ in CONTROL mode, our objective is to verify that $x_1$ will stay in a much more restrictive domain in CONTROL mode: $|x_1| < 3.2$.

We define the unsafe set as $\operatorname{Unsafe}(1)=\emptyset$ and $\operatorname{Unsafe}(2)=\{x\in \mathbb{R}^3| 3.2 \leq |x_1| \leq 10\}$, which is sufficient to prove $|x_1|< 3.2$ in CONTROL mode. Similarly, we tried to use both the method in this paper and the method in \cite{prajna2007framework} to compute the barrier certificate. By setting $\lambda_1=\lambda_2=-\frac{1}{5}$ and $\gamma_{12}=\gamma_{21}=1$, our method found a pair of quartic barrier certificate functions: $\phi_1(x)$ and $\phi_2(x)$, whose zero level set is shown in \figurename~\ref{fig:example31_phi1} and \figurename~\ref{fig:example31_phi2} respectively. As you can see, at each location $l=1,2$, the zero level set of $\phi_l(x)$ forms the boundary of the over-approximation $\phi_l(x) \leq 0$ (denoting the points within the pipe) for the reachable set at location $l$. On the one hand, the hybrid system starts from and evolves within the corresponding over-approximation and jumps back and forth between the two over-approximations. On the other hand, the unsafe set does not intersect the over-approximation formed by $\phi_2(x)\leq 0$ (see \figurename~\ref{fig:example31_phi2_uns}). Therefore, the safety of the system is guaranteed. However, using the method in~\cite{prajna2007framework}, we cannot compute the barrier certificate, which means it cannot verify the system.
\begin{figure}[!t]
 \centering
 \subfigure[$\phi_1(x)=0$]{
  \includegraphics[scale=0.45]{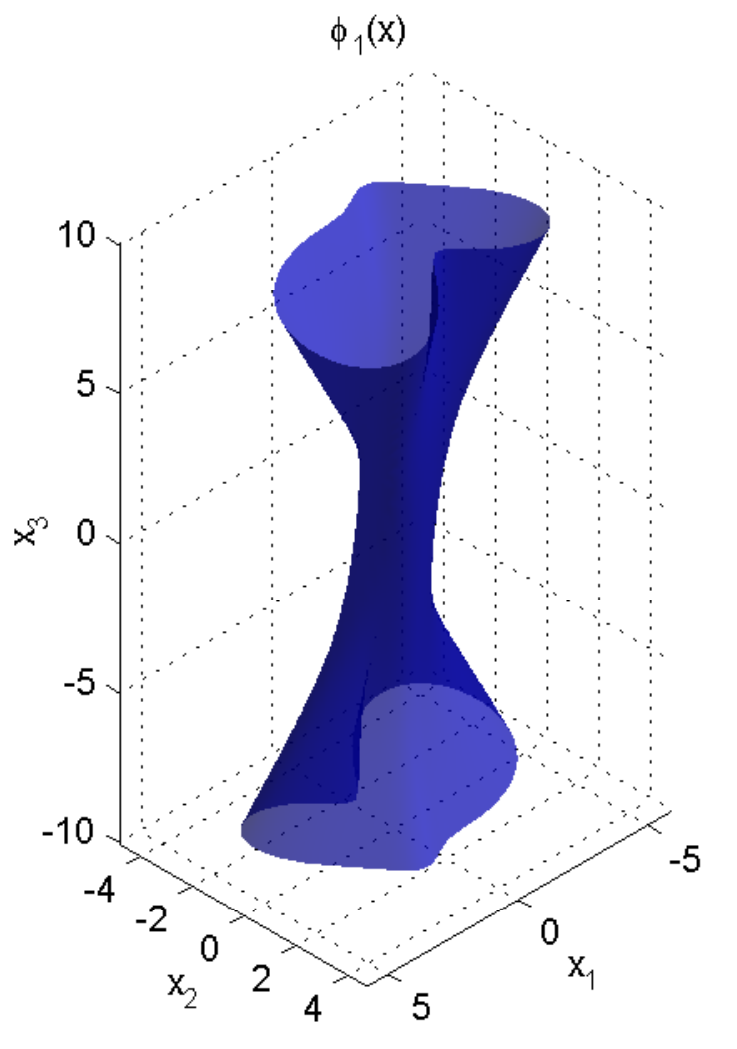}
   \label{fig:example31_phi1}
   }
   \subfigure[$\phi_2(x)=0$]{
  \includegraphics[scale=0.45]{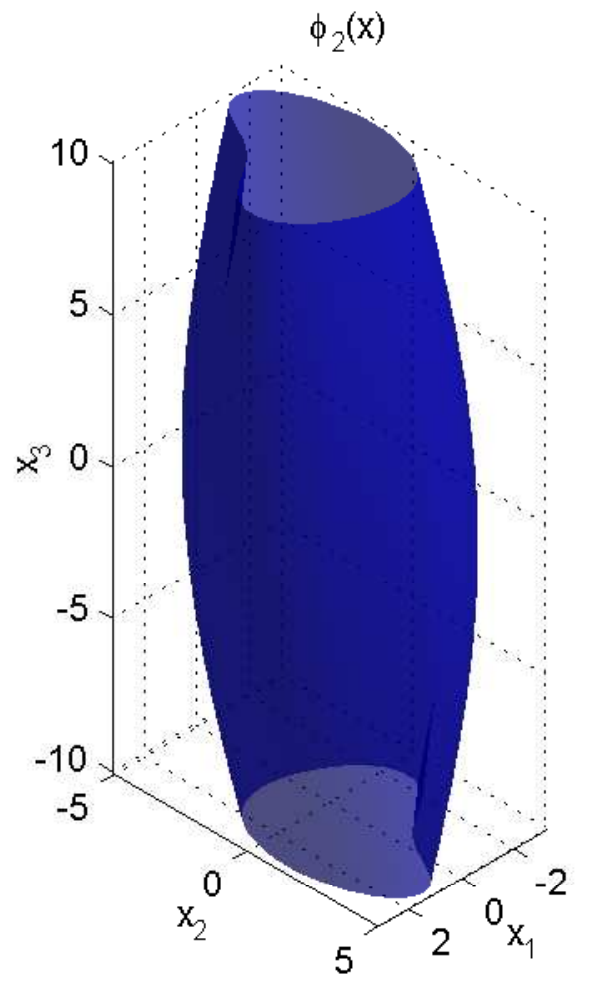}
   \label{fig:example31_phi2}
   }
   \subfigure[$3.2\leq x_1\leq 10$, $\phi_2(x)=0$]{
  \includegraphics[scale=0.45]{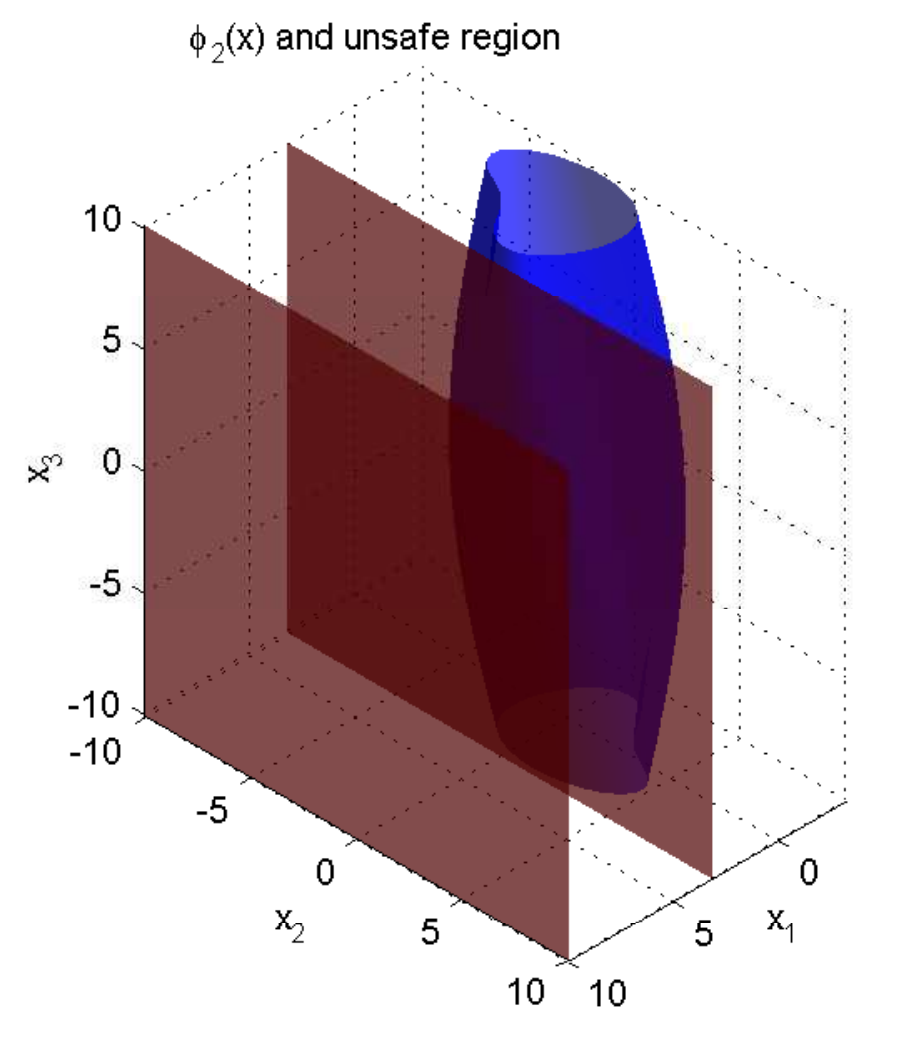}
   \label{fig:example31_phi2_uns}
   }
 \caption{Barrier certificates $\phi_1(x)$ and $\phi_2(x)$ for the hybrid system in Subsection~\ref{subsec:example3}. $\phi_l(x)=0$ ($l=1,2$) forms the boundary of the over-approximation $\phi_l(x)\leq 0$ and separates the inside reachable set from the outside unsafe set (e.g. $3.2 \leq x_1 \leq 10$).}
 \label{fig:barriercert}
\end{figure}

\section{Conclusion}\label{sec:conclusion}
In this paper, we propose a new barrier certificate condition (called \emph{Exponential Condition}) for the safety verification of continuous systems and hybrid systems. Our barrier certificate condition is parameterized by a real number $\lambda$ and the conservativeness of the barrier certificate condition depends closely on the value of $\lambda$: the less value the $\lambda$ is, the less conservative the barrier certificate condition is. Specifically, \emph{Convex Condition} is just the special case of \emph{Exponential Condition} with $\lambda=0$. Therefore, we can obtain the barrier certificate condition that is less conservative than \emph{Convex Condition} as long as we set $\lambda$ to a negative value. The most important benefit of \emph{Exponential Condition} is that it possesses a relatively low conservativeness as well as the convexity and hence can be solved efficiently by semidefinite programming method.

Based on our method, we are able to construct polynomial barrier certificate to verify very critical safety property for semialgebraic continuous systems and hybrid systems. The experiments on a continuous system and a hybrid system show the effectiveness and practicality of our method.


\bibliographystyle{abbrv}


\end{document}